\documentclass[reqno,a4paper]{amsart}
\usepackage{amsmath, amssymb}
\usepackage[sans]{dsfont}  %%%carlangelo
\usepackage{enumerate, xspace}
\usepackage[colorlinks]{hyperref}
\numberwithin{equation}{section}
\newtheorem{thm}{Theorem}[section]
\newtheorem{lem}[thm]{Lemma}

\newtheorem{rem}[thm]{Remark}

\newcommand\N{{\mathbb N}}

\newcommand\E{{\mathbb E}}
\newcommand\R{{\mathbb R}}

\newcommand\1{{\mathds{1}}}   %%carlangelo
\newcommand\HH{{\mathcal H}}
\newcommand\A{{\mathcal A}}
\newcommand\B{{\mathcal B}}

\newcommand\D{{\mathcal D}}

\newcommand\GG{{\mathcal G}}

\newcommand\LL{{\mathcal L}}

\newcommand\SSS{{\mathcal S}}

  \newcommand\z{\zeta}
 
 \newcommand\e{\epsilon}

\newcommand\g{\gamma}

\newcommand\s{\sigma}
\let\a=\alpha
\let\b=\beta

\newcommand{\nada}[1]{}

\newcommand\ks{\boldsymbol k}

\begin{document}
 \title[Spectral properties]
 {Spectral properties of integral operators \\in bounded,  large intervals.}
 \author{Enza Orlandi}\thanks{E.O.  thanks  the anonymous referee for  pointing  out  an error in the computation of the Cheeger's constant  in an earlier version of the article. \\
Sadly, the first author died before finishing the revision of the paper. As the paper was almost complete and C.L. had some knowledge of its content, he decided to finish it in memory of a splendid person.}
\address{Enza Orlandi, 
Dipartimento di Matematica e Fisica\\
Universit\`a  di Roma Tre\\
 L.go S.Murialdo 1, 00156 Roma, Italy. }
 \author{Carlangelo Liverani}
\address{Carlangelo Liverani\\
  Dipartimento di Matematica\\
  II Universit\`{a} di Roma (Tor Vergata)\\
  Via della Ricerca Scientifica, 00133 Roma, Italy.}
\email{{\tt liverani@mat.uniroma2.it}}
\date{\today}
\begin{abstract}
We study  the spectrum of  one dimensional  integral operators in  bounded  real intervals of  length  $2L$, for  value  of  $L$ large. The integral operators    are obtained by linearizing   a non local evolution equation for a non conserved order parameter describing the phases of a fluid.  We prove a  Perron-Frobenius theorem showing that there is an isolated, simple minimal eigenvalue   strictly positive  for $L$ finite,  going   to zero exponentially fast in $L$.
 We   lower bound,    uniformly    on  $L$,  the   spectral  gap  by applying a generalization of the Cheeger's inequality.
These   results  are  needed for deriving spectral properties  for non local Cahn-Hilliard  type of equations  in problems of interface dynamics, see \cite {O2}.
\end{abstract}

\maketitle

\section{Introduction}   We study the spectrum of an integral operator  acting on  $L^2$ functions defined in intervals $[-L,L] \subset \R$,  for  value  of  $L$ large.
This problem arises  when analyzing 
   layered  equilibria and front dynamics 
 for   the  conservative, 
  nonlocal, quasilinear  evolution equation   typified by 
\begin{equation} \label {1.0} \partial_t m(t,x)    =  \nabla\cdot\bigl \{
\nabla  m(t,x)  - 
  {\beta (1-m(t,x)^2) (J\star
\nabla m )(t,x)} \bigr\},  \end {equation} 
where $\beta>1$,
 $$ (J \star  m)(x)= \int_{\R} J (x, y) m(y) {\rm d}y  $$
and  $J(\cdot, \cdot)$  is a regular, symmetric, translational invariant,  non negative function   with compact support and integral equal to one.
This equation   \eqref {1.0} first appeared
in the literature in a paper  \cite{LOP}  on the dynamics of Ising systems
with a long--range interaction and so--called ``Kawasaki'' or ``exchange''
dynamics and  later it was rigorously
derived  in \cite{GL1}. 
 In this physical context,
$m(x,t) \in [-1,1]$ 
  is the spin magnetization density. 
 It has been formally shown by  Giacomin and Lebowitz   \cite{GL2},  that in the sharp interface limit, i.e 
  the limit in which the phase domain  is  very large with
respect to the size of the interfacial region and time is suitably rescaled,   the limit motion is given by Mullins Sekerka motion,
 a quasi-static free
boundary problem in which the mean curvature of the interface plays a fundamental
role.
Equation \eqref  {1.0}  could  be considered as a   non local type of Cahn-Hilliard equation.   Our intention is to  provide basic spectral estimates useful for deriving  higher dimensions spectral results  in order to establish rigorously the relation between \eqref {1.0}  and the  singular limit motion described by the    Mullins Sekerka equations, see \cite {O2}. 
We recall some  previous  results useful to    better  contextualize   the problem.
 When $\beta>1$  there is a phase transition in the underlying spin system, \cite {LP}. The pure phases correspond to the stationary spatially homogeneous solutions  of   \eqref {1.0} satisfying
 $$ m = \tanh  \b m.$$
 For $ \beta>1$ there are three and only three different roots denoted
 $$ \pm m_\beta, 0.$$ 
 The two phases $\pm m_\beta$ are thermodynamically stable while $m=0$ is unstable.  These statements, established in the context of the theory of Equilibrium Statistical Mechanics, see  \cite {LP}, are reflected by the corresponding stability properties of the space homogeneous solution of \eqref {1.0}, see \cite{GL2}.  
 Equations \eqref   {1.0} has also stationary solutions connecting the two coexisting phases: they are all identical modulo translations and reflection,  see  \cite{GL2},   to the "istanton"  $\bar m( \cdot)$ which is $C^\infty (\R)$, strictly increasing, antisymmetric function which identically verifies
  \begin{equation} \label {G2a}
 \bar m(x) = 
 \tanh \b (J \star \bar  m)(x), \qquad x \in \R .
 \end {equation} 
$\bar m(\cdot)$ is the stationary pattern that connects the minus and the plus phases as
 \begin{equation} \label {G3a}
\lim_{x \to \pm \infty}  \bar m(x)= \pm m_\b,   \end {equation} 
 and  it can be interpreted as a diffuse interface. 
 The first results on these stationary patterns were obtained when analyzing  the non conservative equation 
 \begin{equation} \label {G1a}
  \partial_t m(t,x)= -m(t,x)+
 \tanh \b (J \star  m)(t,x).   
 \end {equation} 
  Equation \eqref {G1a} has been derived from the Glauber  (non conservative) dynamic of  an Ising spin system interacting via a Kac potential, see   \cite {DOPT1}.
  Since  both the equations \eqref {1.0} and \eqref {G1a}  have been derived from the same Ising spin systems, the first by a conservative dynamic  while the latter by a non conservative one,   both   have  as equilibrium solutions 
  the homogeneous  solution $\pm m_\beta$ and the  stationary patterns  connecting the two homogeneous phases.
 Stability properties of  $\bar m$ have been derived both for the conservative evolution \eqref {1.0}  (see \cite {CCO1},  \cite {CCO2}, \cite {CO}) and for  the nonconservative evolution \eqref {G1a} (see \cite {DOPTR}). 
 
 Let us  recall few  previous  results which are needed in this paper. 
    As proved in \cite {DOPTR} the interface described by the istanton is ``stable"  for equation \eqref {G1a} and  any initial datum ``close to the instanton" is attracted and eventually converges exponentially  to some translate of the instanton.
Linearizing     the evolution equation \eqref {G1a}  at  $  \bar m$  one obtains 
the integral operator 
\begin{equation} \label {G3} \LL v =  v - \beta (1 -\bar m^2) J \star v   \end {equation} 
which is selfadjoint  when $v \in L^2 ( \R, \frac 1 { \beta (1 -\bar m ^2)} dx)$.
The spectrum of this operator has been studied in 
  \cite {DOPTE}. It  has been proved that the spectrum of $\LL$ is positive, the lower bound of the spectrum is 0
  which is an eigenvalue of multiplicity one and the corresponding eigenfunction is $ \bar m' (\cdot)$, the spatial derivative of the stationary solution $\bar m$, 
  i.e 
 \begin{equation} \label {G4} \LL \bar m'  =0.   \end {equation} 
    The remaining part of the spectrum is strictly bigger  than some positive number. 
  In this paper  we consider operators  of the type of  the operator  $\LL$ defined in \eqref {G3} but acting   over functions in  bounded intervals $[-L,L]$, $L$ large. In fact, these are the operators naturally arising in the study of higher dimensional non local  Cahn-Hilliard type equations.  
  
Our goal is to derive qualitative and quantitative  properties (in terms of $L$)  of the principal eigenvalue together with the associated eigenfunctions  and, most importantly, to show that the second eigenvalue can be bounded uniformly from below in terms of  the  length of the interval. 

Remark that the above informations  cannot be derived  from the knowledge  of the spectrum of $\LL$ 
  in  $L^2 ( \R, \frac 1 { \beta (1 -\bar m ^2)} dx)$, obtained in  \cite {DOPTE}.  Furthermore the method (Fourier analysis and Weyl's Theorem) used in  \cite {DOPTE} is not applicable to the present context, hence the need to develop a different, more flexible, approach. 
  
It might be helpful to compare heuristically what we are doing  with  similar problems  analysed previously  in the context of  reaction diffusion equations and Cahn-Hilliard equations.
 Dividing by  $\beta (1 -\bar m ^2)$   the operator $\LL$  we can define a new operator 
  $$ \GG v =  \frac  v { \beta (1 -\bar m^2)} -  J \star v =  - \left [ J \star v -v\right ]  +  f''(\bar m) v $$
  where $$f''(\bar m) = -1 +\frac  1 { \beta (1 -\bar m^2)} $$
  and   
  $$ f(m) = -  \frac 12  m^2 + \frac 1 \beta \left [ \left  (\frac {1+   m} 2 \right )\ln \left (\frac {1+   m} 2 \right)+   \left (\frac {1 -  m} 2 \right )\ln \left (\frac {1-   m} 2 \right) \right ] 
$$
 is a double  equal well potential. 
 The operator $\GG $ on $L^2 ( \R, dx)$ and the operator $\LL$  on  $L^2 ( \R, \frac 1 { \beta (1 -\bar m ^2)} dx)$  have the same spectrum. 
 Assume that $v$ is smooth, taking into account that  $J$  is symmetric  and therefore the first  moment is null,
 we have that $J \star v -v \simeq \Delta v$.
 Heuristically $  - \left [ J \star v -v\right ]  +  f''(\bar m) v$  is equivalent  to   $ -\Delta v  +  f''(\bar m) v.$
 So the problem we are dealing with is in the same spirit of the problem dealt by  De Mottoni and  Schatzman, see   \cite [subsection 5.4] {deMS}.  They  studied the spectrum of  $ -\Delta v  +  W(\bar \theta) v$  in  the  finite interval $[-L,L]$ with Neumann boundary conditions.  We denoted by $W(\bar \theta)$ the corresponding of  $f''(\bar m)$ in  \cite {deMS}. This was a basic result used to obtain higher dimension spectral results for the Cahn-Hilliard equations, see for example \cite{Chen} and \cite {AF}.
 In this paper we establish results for the spectrum of one dimensional integral operator in  the  finite interval $[-L,L]$.
 The  main difficulty  is to show that the spectral  gap   of our integral operator
is  bounded uniformly  in $L$.  This is achieved by applying a generalization of Cheeger's inequality, proven in \cite {LS} and  lower bounding in our context  the Cheeger's constant.

\section { Notations and Results} 

Let $T_L =[-L,L]$ be a real interval, 
  $L  \ge 1 $.    We are actually  interested  in   $L$  large. 
 \subsection {The interaction}  \label{sec:interaction}
Let $J \in C^\infty(\R,\R)$,  be a symmetric probability kernel, that is: $J(x)=J(-x)$, $J\geq 0$ and  $ \int J (x) {\rm d}x =1$.  We assume that  $ J(x)>0$ for all $x\in (-1,1)$ and $J(x)=0$ for $|x|\geq 1$.\footnote{ Note that the choice of $[-1,1]$ as the support of $J$, i.e. the interaction length, is tantamount to a choice of the unit of measure. Indeed, the problem is scale invariant and the only quantity that really matters is the ratio between the range of the interaction (in this case 1) and the length of the interval (in this case $L$).} To    define  the interaction between   $x$ and $y$ in $\R$  we   set, by an abuse of notation, $J(x,y) = J(y-x)$.    For a function $v$ defined on  $T_L$  we  set 
    \begin{equation}
 \label{tm1}  (J \star_b v) (x) = \int_{T_L} J(x-y) v(y) {\rm d}y. \end{equation} 
 The suffix $b$ is to reminds the reader that the integral is on  the  bounded interval $T_L$.
       Notice    $ \int_{T_L} J  (x,y) dy = b(x)$ with $ b(x) \in [\frac 12, 1]$ for $x \in T_L$.
There are  other     ways    to  derive from $J$ an integral kernel acting only on  functions on the bounded interval $T_L$.   One  is the  following
   \begin{equation}
 \label{tm1a} J^{neum} (x,y)= J(x,y) + J(x, 2L-y) + J(x, -2L-y),  \end{equation} 
 where $2L-y$  is the image of $y$ under reflection  on the  right boundary $\{L\}$ and 
 $-2L-y$  is the image of $y$ under reflection  on the  left  boundary $\{-L\}$.
   By the assumption on $J$,     $ J^{neum}(x,y)=J^{neum}(y,x)$ and  $ \int J^{neum} (x,y)  {\rm d}y=1$ for all $x \in T_L$.
The choice  to define by boundary reflections the interaction \eqref {tm1a} has  the advantage  to keep $J^{neum}$   a symmetric probability kernel.   This definition  first appeared in the paper  \cite [Section 2] {DOP}  and it was  called there 
  ``Neumann" interaction.  In   \cite {DOP} the authors studied   spectral properties of  operators  closely related to the operator  $\LL$, see \eqref {G3},  defined on the space  of the continuous symmetric functions on $\R$, $ C^{\rm {sym}} (\R)$.         
   
 We will  consider  in this paper  operators    with the integral kernel   \eqref {tm1}  acting on Hilbert spaces.
 We       could   denote    \eqref {tm1}   the Dirichelet interaction  kernel.      Our results  can be, with minor modifications,  immediately extended to the case  when  the integral kernel is  $J^{neum}$.  
   
 \subsection {The  istanton}
 We  call istanton  the  antisymmetric  solution        $\bar m$  of \eqref {G2a} with conditions at infinity given in \eqref {G3a}. The function  $ \bar m \in C^\infty (\R)$, it  is strictly 
increasing,  and  there exist $a>0$, $\alpha_0> \alpha >0$ and $c>0$ so that 
\begin{equation} \label{decay} \begin {split}
&0  <  m_\beta^2 - \bar  m ^2(x)  \le ce^{-\alpha |x|}\ ,  \cr  
  &  \quad | {\bar m}'  (x)- a \a  e^{-\alpha |x|}| \le    ce^{-\alpha_0 |x|}. 
\end {split}\end{equation} 
A proof of these estimates      and     related  results  can be found in   \cite[Chapter 8, Section 8.2]{Pr}.

\subsection {The Operator} For  $\beta>1$ 
set 
 $p(x)= \beta (1-\bar m^2(x)
)$  
 where $\bar m$ is the    istanton.   By the   properties of  $ \bar m $  we have  that 
 \begin {equation} \label {S.10a} \lim_{|x| \to \infty} p(x)= \beta (1-m^2_{\beta}) <1,  \end {equation}
 and 
\begin {equation} \label {S.10}\beta  \ge  p(x) \ge \beta (1-m^2_{\beta}) >0, \qquad x \in \R. \end {equation}
Denote  
$$ \HH =  L^2(T_L, \frac 1 {p(x)}  {\rm d}x), $$
and  for $v \in \HH$ and $w \in \HH$
$$ \langle v, w \rangle = \int_{T_L} v (x) w(x)\frac 1 {p(x)} {\rm d} x ,$$
 \begin {equation} \label {E.1} \|v\|^2 = \int_{T_L} v^2(x)\frac 1 {p(x)} {\rm d} x. \end {equation}
 To stress the dependence   of $\HH$  on $L$ we will   add, when needed,  a suffix $L$,  writing $\HH_L$.
   We   denote by 
 $$\|v\|_2, \qquad  \|v\|_\infty,$$
 respectively the $L^2 (T_L, dx) $  and the $L^\infty (T_L, dx) $ norm  of a   function $v$.  
By \eqref {S.10} we have 
\begin {equation} \label {E.1a} 
\|v\|^2_2 \le \beta \|v\|^2. 
\end {equation}

We can now define precisely the operator we want to study: let    $\LL^0$ be  the operator   acting  on $ \HH $   as
\begin{equation} \label{op1} 
(\LL^0 g) (x) =  g(x)-  p(x) ( J \star_b g) (x).    
\end {equation}

\subsection {Results}
 The following  results for the operator $\LL^0$ hold for any fixed value of $L$ large enough.

\vskip0.5cm
\begin {thm} \label {81}   For any $\beta>1$ there exist $L_1(\beta)$ so that for $L\ge L_1(\beta)$ the following holds. 

(0)The operator $\LL^0$ is a bounded, quasi compact,  selfadjoint  operator on $\HH$.

(1)  There exist  $ \mu_1^0  \in \R $ and  $ \psi^0_1 \in \HH $, $ \psi^0_1 $  strictly positive in $ T_L$  
so that 
$$  \LL^0  \psi^0_1=   \mu_1^0  \psi^0_1.$$
The eigenvalue $ \mu_1^0 $  has multiplicity one and any other point of the spectrum is strictly
bigger than  $\mu_1^0 $.  There exist $c>0$  independent on $L$ so that 
   \begin{equation} \label{8.2}    0 \le  \mu_1^0  \le c e^{-  2\a  L},   \end {equation}
   where $\a>0 $ is  given in \eqref {decay}.
Further  $\psi^0_1 \in C^\infty (T_L)$, $\psi^0_1(z) =  \psi^0_1(-z)$ for $z \in T_L$.

(2) Let $ \mu_2^0 $ be the second  eigenvalue   of $\LL^0$.    
We have that 
  \begin{equation} \label{8.3} \mu_2^0= \inf_{ \langle \psi,\psi^0_1\rangle =0; \| \psi\|=1} \langle \psi, \LL^0 \psi \rangle \ge D, 
 \end {equation}
 where $D>0$ independent on $L$  is given in \eqref {LR.1}.

(3)  Let  $\psi^0_1$ be the
  normalized eigenfunction  corresponding to $\mu_1^0$ we have
  \begin{equation} \label{8.4} 
  \left\| \psi^0_1- \frac {\bar m'}{\|\bar m'\|}\right\|  \le  C  e^{-
2 \a  L}, 
 \end {equation}
 where $C>0$ is a constant independent on $L$.
 \end {thm}
 \section {Proof of the results}
 To  prove    Theorem   \ref {81} we introduce  the following  auxiliary
 operators. 
 Denote by $\A$ the linear integral operator acting on functions $g
\in \HH$
\begin {equation}   \label {S1.2}\A g (x)=   p(x) (J \star_b g) (x).  \end {equation}
We  denote by $\B$  the operator acting on  $L^2(\R, \frac 1 {p(x)} dx)$:
 \begin {equation}   \label {S1.2a}\B g (x)=   p(x) (J \star g) (x).  \end {equation}
 The operator $\B$ has been studied in  \cite {DOPTE}  and we will use that, recall \eqref {G4},
 \begin {equation}   \label {m1} \bar m '(x) = (\B \bar m') (x).  \end {equation}
 We have the following result.
 
\vskip0.5cm
\begin {thm} \label {P-F}   Take $L\ge 1$. The operator  $\A $ is  a compact, selfadjoint operator  on   $\HH$,
positivity improving.
Further, there exist  $\nu_0>0$ and  $v_0 \in \HH $, $v_0 $  strictly positive even function,  so that
\begin {equation}   \label {S1.1} \A v_0(x) = \nu_0 v_0(x) \qquad  x \in T_L. \end {equation}
The eigenvalue $ \nu_0$  has multiplicity one and any other point of the spectrum is strictly
inside the ball of radius $\nu_0$. The eigenfunction   $v_0$  is    in   $  C^\infty ( T_L)$.  
\end {thm}
 \begin {proof} 
 It is immediate to   see that
 $$ \langle  \A v, w \rangle =  \langle    v, \A w \rangle.$$
The compactness can be shown  by proving that  any bounded set of $\HH$ is mapped by $\A$  in a relatively compact 
set.  Namely since $J (\cdot, \cdot)$ is continuous in $T_L\times T_L $ and $T_L$ is compact, then $J (\cdot, \cdot)$  is uniformly continuous. Thus given $\e>0$, we can find $\delta>0$ so that  $|x-y|\le \delta$ implies $|J(x,z)-J(y,z)| \le \e$ for all $z \in T_L$.  The same holds for $p(\cdot) J(\cdot, z)$.  Let  $B_M = \{ v \in \HH: \|v\|^2 \le M \}$.  If $ v \in B_M$ and $|x-y|\le \delta$ we have  
$$ |(\A v)(x)-(\A v)(y)|  \le  \e   c(\beta, J)  \|v\| \le  \e c(\beta, J) M,$$
where $c(\beta, J) >0$ depends only on $\beta$ and $J$. 
Therefore the functions $\A [B_M]= \{w \in \HH : w= \A v,  \, v \in B_M \}$ are equicontinuous. Since they are also uniformly bounded by $c(\beta) \|J\|_2 M$,  where $c(\beta)>0$, we can use the Ascoli theorem to conclude that for every sequence $\{v_n\} \in B_M$, the sequence $\{\A  v_n\}$  has a convergent subsequence (the limit might not be in $\A [B_M]$) in $C[T_L]$ and therefore in $\HH$. 
To show the positivity improving we  must show that there exists $n_L\in\N$ such that, for all $v(z) \ge 0$, $ v \not\equiv 0$,  we have
$(\A^{n_L} v) (z) >0$. 
To see it, note that $\A^nv(x)=\int_{T_L} K_n(x,y)v(y){\rm d} y$ where
\[
K_n(x,y) = p(x)  \int  {\rm d} x_1 {\rm d} x_2....{\rm d} x_n p(x_1)  J (x,x_1)p(x_2)  
J (x_1,x_2)... p(x_n)  J (x_{n},y).
\]
On the other hand, by assumption, there exists $\omega>0$ such that $\inf_{|x-y|\leq 1/2}J(x,y)\geq \omega$.
Next, we prove by induction that there exists $\zeta_*>0$ such that
\[
\inf_{|y-x|\leq n/4}K(x,y)\geq \zeta_*^n.
\]
This is obviously true for $n=1$ and $\zeta_*\leq \omega(1-m^2_{\beta})$. Let us assume it true for $n$, then
\[
K_{n+1}(x,y)=p(x)\int_{T_L} J(x,x_1)K(x_1,y)\geq  \beta (1-m^2_{\beta})\zeta_*^n\int_{T_L\cap \{|x_1-y|\leq n/4\}}J(x,x_1).
\]
If $|x-y|\leq n/4$ then
\[
K_{n+1}(x,y)\geq \frac 12 \beta (1-m^2_{\beta})\zeta_*^n
\]
while, if $|x-y|\in [n/4, (n+1)/4]$, then
\[
K_{n+1}(x,y)\geq \beta (1-m^2_{\beta})\zeta_*^n\int_{T_L\cap \{|x_1-y|\in [(n-1)/4, n/4]\}}J(x,x_1)\geq \frac \omega 4 \beta (1-m^2_{\beta})\zeta_*^n,
\]
from which the statement follows with $\zeta_*=\frac\omega4 \beta (1-m^2_{\beta})$. Accordingly, setting $K(x,y)=K_{8L}(x,y)$ and $\zeta=\zeta_*^{8L}$ we have
\begin {equation}   \label {S1.3}  
K(x,y)  > \z\quad\quad\textrm{ for all } x,y\in T_L .
\end {equation}    
Thus the operator $\A$ is positivity improving and  moreover one can apply the classical Perron Frobenius Theorem to  the kernel $K (\cdot,\cdot)$.
 As a consequence we have that $ \A^{8L}$ has a simple strictly positive maximal eigenvalue and a spectral gap.
Accordingly, the maximum eigenvalue of the spectrum of $\A$, which we denote $ \nu_0$,  has   multiplicity one and any other point of the spectrum of $\A$  is strictly smaller than  $\nu_0$.   Further the   eigenfunction  associated to $\nu_0$ does not change sign. So we assume that it is 
 strictly positive and we denote it  $v_0$.
Next we show that $v_0$ is even.  Denote by $ w(x)= v_0(-x)$.
Since $J$ and  $ p(\cdot)$ are   even functions we have that
$$ (\A w)(x) =  (\A v_0)(-x)  = \nu_0  v_0(-x) =  \nu_0 w(x). $$
We then deduce that the function $w$ is an eigenfunction associated to $ \nu_0$.   Since $\nu_0$ has multiplicity one
we must have that $w(x) =  v_0(x)$.  Therefore  $v_0$ is even.
Next we show that $v_0 \in C^\infty (T_L)$.  We start proving that it is $C^1(T_L)$.
Since  $p(\cdot)$ is $C^\infty (\R)$ and  $J \in C^\infty(\R)$ differentiating we obtain
 \begin {equation}   \label {a.6}  \nu_0 v_0'(z) = p' (z) (J \star_b v_0)(z) + p(z)  (J \star_b v_0)'(z). \end {equation}
    Therefore  $v_0 \in C^1(T_L)$.
 Since  $(J  \star_b v_0)'(z) = (J'  \star_b v_0)(z)$,
 differentiating again \eqref {a.6}  we can show that $v_0 \in C^2(T_L)$.  Repeating  the argument  yields  $v_0 \in C^\infty (T_L)$.
 \end {proof}

\vskip0.5cm
\begin {lem} \label {S11}  {\bf (Lower bound on    $\nu_0$)}      There exists  positive
constant
$c>0$   independent on $L$ so  that  for any $L\ge 1$
  \begin{equation} \label{S1.5} \nu_0 \ge 1 -c  e^{ -2\a    L},
 \end {equation}
 where $\a>0$ is the constant in \eqref {decay}.
\end {lem}
\begin {proof} Consider the following trial function
 \begin{equation} \label   {S11.1} h(x)=   \frac {\bar m '(x)} {\|{\bar m '} \|},   \qquad  -L  
\le x
\le L.  \end {equation}
By the variational formula for spectral radius we have that
$$ \nu_0 \ge \langle \A h, h \rangle,$$
hence
 \begin{equation} \label{S11.2} 
  \begin {split}   \langle \A h, h \rangle   & = \int {\rm d} x  \frac 1 {p(x)} h(x)  p(x)    (J \star_b h) (x)  \cr 
  & =    \frac {1} {\|{\bar m '} \|^2} \int {\rm d} x  \frac 1 {p(x)} \bar m' (x) (\B \bar m') (x) \\
  &\phantom{=} +    \int {\rm d} x  \frac 1 {p(x)} h(x)  p(x)   \left [ (J \star_b h) (x)  - \frac {1} {\|{\bar m '} \|}(J \star \bar m') (x)\right ],
\end {split}   
\end {equation}
 where $\B$ is the operator defined in \eqref {S1.2a}.  
By   \eqref {m1}  and \eqref {E.1} we have that 
$$  \frac {1} {\|{\bar m '} \|^2}  \int {\rm d} x  \frac 1 {p(x)} \bar m' (x)  (\B \bar m') (x)   =  1. $$
Further, since $\bar m' $ is even,   we have 
\begin{equation} \label{S11.2a} 
\begin{split}
&\int {\rm d} x    h(x)     \left [ (J \star_b h) (x)  - \frac {1} {\|{\bar m '} \|}(J \star \bar m') (x) \right ]  \\
&= 2 \int_{L-1}^L  {\rm d} x   h(x)      \left [ (J \star_b h) (x)  - \frac {1} {\|{\bar m '} \|}(J \star \bar m') (x)  \right].       
\end{split}
\end {equation}
 For $x \in [L-1,L]$  we have  
 $$ (J \star_b h) (x)  - \frac {1} {\|{\bar m '} \|}(J \star \bar m') (x)   =  \frac {1} {\|{\bar m '} \|}  \left [ \int_{x-1}^L  {\rm d} y J (x,y) \bar m'(y)- \int_{x-1}^{x+1}  {\rm d} y J (x,y) \bar m'(y)\right ].  $$
Let   $c$  be a positive constant independent on $L$,  we obtain 
 \begin{equation} \label{Mr1}  
 \begin {split}  
 \int_{x-1}^L  {\rm d} y J (x,y) \bar m'(y)- \int_{x-1}^{x+1}  {\rm d} y J (x,y) \bar m'(y)&   =   - \int_{L}^{x+1}  {\rm d} y J (x,y) \bar m'(y)\\
 &    \ge - c e^{-\a L},
    \end {split}   
 \end {equation}
 since    $\bar m'(\cdot)$ is   strictly positive and   exponentially decreasing, see \eqref {decay}. Inserting \eqref {Mr1} in \eqref  {S11.2a} we obtain the  lower bound \eqref  {S1.5}. 
\end {proof}
\begin {lem} \label {S14}  {\bf (Upper bound on    $\nu_0$)}  We have that  for any $L\geq1$   
\begin {equation} \label {S1.10}   \nu_0 <1.  \end {equation}
 \end {lem}
 \begin {proof} 
 Multiply  \eqref {S1.1} in $L^2 (T_L,\frac 1 { p(x)} dx)$  with the trial function $h$ introduced in \eqref {S11.1}.
We have
$$  \nu_0 \langle v_0,h \rangle  =  \langle \A v_0,h\rangle =   \langle  v_0, \A h\rangle.   $$
By \eqref {m1} write
\[
\A h = \B \frac {\bar m'}{\|\bar m'\|}  +   p \left[J \star_b h - J \star h\right]
=h  +   p \left[J \star_b h - J \star  h\right]. 
\]
Then  
\begin {equation} 
\begin {split}  
& \langle  v_0,    p [J \star_b h - J \star  h]\rangle 
=   \int_{-L}^{L} \frac {v_0(x)}  {\|\bar m'\|} [(J \star_b   \bar m' )(x) - (J \star   \bar m')(x) ]   {\rm d} x  \cr 
& =        \frac {2} {\|{\bar m '} \|} \int_{L-1}^L  {\rm d} x   v_0 (x)      \left [ (J \star_b  \bar m') (x)  - (J \star \bar m') (x)  \right]   \cr &=
  -  \frac { 2}  {\|\bar m'\|}     \int_{L-1}^L  {\rm d} x   v_0 (x)   \int_{L}^{x+1}  {\rm d} y J (x,y) \bar m'(y)  <0 
 \end {split}  
\end {equation}
since  $v_0 $ is an even positive function, see Theorem \ref{P-F}. It follows
\[
\nu_0 \langle v_0,h \rangle < \langle  v_0, h \rangle  
\]
Since both $\bar m'$ and $v_0$ are positive  we have
$ \langle v_0,h \rangle  >0$, which proves the Lemma.
\end {proof}

  \begin {rem}  \label {roma1}   It is easy to verify that if  $\A$ is defined  replacing $J\star_b$ with  the integral kernel $ J^{neum} \star$   we will have  $ \nu_0  \ge 1 + ce^{-2 \a L}$.  \end {rem} 

Next we investigate the properties of the maximal eigenvector $v_0$.

\vskip0.5cm
\begin {lem} \label {S13} {\bf  (Properties of     $v_0$)}  
  For any $\e_0 \in (0,  \frac {(1-\s(m_\b))} {2})$  there exists  $r_0= r_0(\e_0)$ and   $L_1= L_1(\e_0)>0$     so that the following holds.  Take $L \ge L_1$ and  let $v_0$ be  the  strictly positive normalised
eigenfunction  of   $\A$ on  $ \HH_L$ corresponding to
$\nu_0$, see Theorem \ref {P-F}.   We have that~\footnote{ By $\lfloor x\rfloor$ we mean the integer part of $x$.} 
\begin{equation} \label{S1.6} 
v_0(x) \le   e^{-\lfloor |x|-r_0\rfloor\a (\e_0)}\sup_{|y|\in[r_0,r_0+1]}v_0(y) \qquad |x| \ge r_0+1 
\end {equation}
where $\a (\e_0)$ is given in \eqref {S12.4}
 \begin{equation} \label{S1.8a} |v_0 ' (x)| \le C , \qquad    x \in T_L, \quad \hbox { with $C$ independent on $L$;}
\end {equation}
 \begin{equation} \label{S1.8} \gamma \ge  \frac {v_0 (x)} {v_0 (y)} \ge  \frac 1 \gamma
  \qquad \qquad  |x-y| \le 1, \quad  \hbox {for any}\quad  x, y \in T_L  \end {equation}
where $\gamma= \gamma (\e_0)>1$ is defined in \eqref {S.2}. 
There exists $ r_1>0$ and  $\z_1>0$ independent of $L$ so that 
 \begin{equation} \label{S1.9} v_0(x)  \ge  \z_1 \qquad \hbox {for }\quad  |x| \le r_1. \end {equation}
  \end {lem}
  \begin {proof}    Fix  $ \e_0 \in  (0,  \frac {(1-\s(m_\b))} {2})$  and let  $r_0= r_0(\e_0) $  be so that 
 \begin{equation} \label{roma3} 
 p(x) <1- \e_0, \qquad |x| \ge r_0. 
 \end {equation}
  Take   $L_0(\e_0)$  so that   for $ L \ge L_0(\e_0)$, $r_0 \le \frac  L 2 $. 
   Furthermore, choose  $L_2(\e_0) $ so large so  that for $L\ge L_2(\e_0)$,   see    Lemma \ref {S11},   
   \begin{equation} \label{roma4} 
   \nu_0 \ge  1   -c  e^{ -2\a    L}>   1- \frac {\e_0} 2. 
   \end {equation}  
    Set $L_1(\e_0) = \max \{ L_0(\e_0), L_2(\e_0)\}$. 
  Then for $L \ge  L_1(\e_0)$, and $|x|\geq r_0$,
\[
\frac {p(x)}{\e_0}\leq e^{-\a (\e_0)}<1,
\]
 where 
 \begin{equation} \label{S12.4}
 \a(\e_0) = \ln \left (\frac{1- \frac {\e_0}2} {p(r_0)}\right)>0. 
 \end {equation} 
Then, for each $|x|\geq r_0$,
\[
v_0(x)\leq e^{-\a (\e_0)}\int_{x-1}^{x+1}J(x,y) v_0(y) \leq e^{-\a (\e_0)}\sup_{y\in[x-1,x+1]} v_0(y) .
\]
Iterating the above, for $|x|\geq r_0+1$, yields
\[
 v_0(x) \leq e^{-\lfloor |x|-r_0\rfloor\a (\e_0)}\sup_{|y|\geq r_0} v_0(y) .
\]
Note that the sup on the right hand side of the above equation cannot be attained in $(-\infty,-r_0-1)\cup[r_0+1,\infty)$, otherwise, letting $x_*$ be the place in which the sup is attained, one would get $v_0(x_*)\leq e^{-\a (\e_0)} v_0(x_*)$ which is impossible since $v_0>0$. Equation \eqref {S1.6} follows.

Next we show \eqref {S1.8a}.   By the eigenvalue equation 
\begin {equation}   \nu_0 v_0' (x) =   p'(x) (J \star_b v_0 ) (x) +    p(x) (J' \star_b v_0 ) (x). \end {equation}
Therefore, by \eqref {decay} and \eqref {S.10}, \eqref {roma4},  \eqref {E.1a} and Theorem \ref{P-F} we obtain
\begin {equation}  
 \begin {split}       
  |v_0' (x)| & \le  \frac  1 {\nu_0}\left [   2 \beta  m_\b \|\bar m'\|_\infty \|J \star_b v_0\|_\infty +     \|J' \star_b v_0\|_\infty \right ] \cr & \le  \frac  1  {1- \frac  {\e_0} 2 }\left [ 
 2  \beta m_\b \|\bar m'\|_\infty \|J \|_2  \|v_0\|_2 +       \|J'\|_2 \| v_0\|_2 \right ] \\
 &=    \frac  {\beta^{\frac 12}}  {1- \frac  {\e_0} 2 }\left [ 2  \beta m_\b \|\bar m'\|_\infty \|J \|_2 +       \|J'\|_2 \right]    \equiv C. 
 \end {split}    
\end {equation}
Now we show \eqref {S1.8}.  Arguing as in \eqref {S1.3}   there exists $n$, independent on $L$, such that there exists   $ \zeta>0$ so that for any $|x- x'|\leq 3$  in $T_L$,   $(J)^n (x, x') \ge \zeta$. Then, taking into account that $\nu_0 \le 1$, see Lemma \ref {S14},  
since $|x-y|\leq 1$, we have 
 \begin {equation} 
 \begin {split} x
  v_0 (x) &= \frac 1 {\nu_0}  p(x) (J \star_b v_0 ) (x) \\
  &= \frac 1 {\nu_0^n}  p(x)  \int_{T_L \times_n T_L} {\rm d} x_1 \dots   {\rm d} x_n p(x_1) \dots p(x_{n-1}) J(x,x_1) \dots  J (x_{n-1}, x_n) v_0(x_n) \cr 
  &\ge    \beta ^n (1- m^2_\b)^n \zeta \int_{ \max \{y-1, -L\} }^{\min \{y+1,L\}}  {\rm d} x'  v_0(x'). 
 \end {split}
\end {equation}
On the other hand
\begin {equation}  \label{S.1}  
\begin{split}
v_0 (y)  &= \frac 1 {\nu_0}  p(y) (J \star_b v_0 ) (y)  \le \frac \beta {\nu_0} \int_{T_L}  {\rm d} x'  J (y, x') v(x')\\
& \le   \frac \beta {\nu_0}    \|J \|_\infty  \int_{ \max \{y-1, -L\} }^{\min \{y+1,L\}}  {\rm d} x'  v_0(x').  
\end{split}
\end {equation} 
 Therefore,  by \eqref {roma4},   for $L >L_1$,  there exists $\g = \g (\e_0)>1$ so that 
\begin {equation}  \label{S.2}  
\frac {v_0 (x) } {v_0 (y) }  \ge  \frac { \beta^n(1- m^2_\b)^n \zeta  \nu_0 }   { \beta   \|J \|_\infty}   \ge  \frac { \beta^{n-1}(1- m^2_\b)^n \zeta  (1- \frac {\e_0} 2)  }   { \|J \|_\infty}    \equiv \frac 1 \gamma. 
\end {equation} 
The other inequality in \eqref {S1.8} is equivalent to the one proved. 
Next,  we show \eqref {S1.9}.  Since $v_0(x) > 0$ 
certainly $ v_0 (x) \ge c_L>0$ for  $x \in T_L $.
We would like to show that there exists an interval  
  independent on $L$ so that   for $x$ in such an interval ,  $ v_0 (x) \ge  \z >0$ with $\z>0$  independent on
$
L$.   This is shown exploiting that     $v_0$ is exponentially decreasing  for $ |x| \ge r_0$, see  \eqref {S1.6}.   Since 
$$ \|  v_0\|  = 1 $$
we must have that there exists $r_1>0$, independent on $L$,  so that 
\begin {equation} \label {ab1} \int_{- r_1}^{r_1}  \frac 1 {p(x)}( v_0(x))^2 {\rm d} x \ge \frac 1 2. \end {equation}  
Since 
\begin {equation} \label {ab2} \int_{- r_1}^{r_1}  \frac 1 {p(x)}( v_0(x))^2 {\rm d} x \le  \frac 1 {\b(1-m^2_\b)}   \int_{- r_1}^{r_1}    v_0^2(x) {\rm d} x 
 \end {equation}
we obtain, from \eqref {ab1} that
\begin {equation} \label {ab3} \int_{- r_1}^{r_1}    v^2_0(x) {\rm d} x  \ge    \frac 1 2\b(1-m^2_\b)  \equiv c_1^2. \end {equation}
Then, there exists $\bar x\in [-r_1, r_1]$ so that    $v^2_0(\bar x)  \ge   \frac 1 {2 r_1} c_1^2$.
By  \eqref {S1.8a} there exists $\e \in (0,1)$ independent on $L$ so that 
 \[
 v^2_0(x)  \ge \frac 1{4 r_1} c_1^2, \quad  x \in (\bar x-\e,\bar x+\e).
 \]
  Similarly as proved in  \eqref {S1.8}, one  can show that given $\bar r > r_1$ there is $k_{\bar r}$, so that for any $k\ge k_{\bar r}$, there exists    $\zeta^*>0$    so that for any $x \in (-\bar r, \bar r)$ and $y \in (-\bar r, \bar r)$ 
 \[ 
 J^k(x,y) \ge \zeta^*.
 \]
 Then, take $x \in  [-\bar r_1, \bar r_1]$, we get 
\[
\begin {split}  
 v_0 (x) &= \frac 1 {\nu_0}  p(x) (J \star_b v_0 ) (x) \cr & \ge  \frac 1  {\nu_0^k} p(x)
\int_{-r_1}^{r_1} {\rm d} x_1\dots \int   {\rm d} x_k p(x_1)\dots p(x_{k-1})J(x,x_1) \dots  J(x_{k-1},x_k) v_0(x_k)
\cr & \ge \left ( \frac {\b(1-m_\b^2)}  {\nu_0}\right )^k \int_{-r_1}^{r_1}  {\rm d} x_k v_0(x_k)J^k(x,x_k) \cr & \ge 
\left ( \frac {\b(1-m_\b^2)}  {\nu_0}\right )^k \int_{\bar x- \e}^{\bar x+ \e} {\rm d} x_k v_0(x_k)J^k(x,x_k)\\
&  \ge \zeta^*  \left ( \frac {\b(1-m_\b^2)}  {\nu_0}\right )^k    \e \frac 1 {\sqrt r_1} c_1 \equiv \zeta_1.
   \end {split}
\]

 \end {proof}

Theorem \ref {P-F}   shows that for any fixed $L$ the operator $\A$ has a spectral gap, which 
might depend on $L$.  We want to prove that the  spectral gap can be lower bounded  uniformly with respect to $L$.  We achieve  this  following closely the paper  of Gregory  Lawler and Alan Sokal, \cite {LS}.  We apply a generalisation  of  the  Cheeger's inequality for positive recurrent continuous time jump processes and  estimated  the   Cheeger's constant in our context.  
Denote by 
\begin {equation} \label {S1.12}  
Q (x,y) = \frac { p(x)} {\nu_0} J (x,y) \frac {v_0(y)} {v_0(x)} \qquad x, y \in T_L  
\end {equation}
and consider the    operator  
\begin {equation} \label {S1.12a} 
( Qf)(x)= \int_{-L}^L  Q (x,y) f(y)   {\rm d} y  
\end {equation}
for $f  \in L^2(T_L,\pi(x) dx )=L^2(T_L,\pi)$,\footnote{ By a slight abuse of notation we will use $\pi$ both to designate the measure and its density, as this does not create any ambiguity.} where 
\begin {equation} \label {S1.12b}  \pi (x) = \frac {v_0^2 (x) }{p(x)}.  \end {equation}
The operator $Q$ is selfadjoint in $L^2(T_L,\pi )$, it is  a positivity-preserving linear contraction on
$L^2(T_L,\pi )$.\footnote{ In fact on all the spaces $L^p(T_L,\pi)$.}  
 The constant function $1$ is an eigenfunction of $Q$  with eigenvalue 1. 
 Denote by $\Xi$ the map from  $L^2(T_L, \frac 1 {p(x)} dx)$  to $L^2(T_L,\pi(x) dx)$, 
 so that 
 \[
  \Xi f = \frac f {v_0}
 \] 
 The map $\Xi$ is an isometry, $\|f\|_{L^2(T_L, \frac 1 {p(x)} dx)} =   \|\Xi f\|_{L^2(T_L, \pi (x)  dx)}$,
 and  
 \[
 \nu_0 Qf = \Xi  \A \Xi^{-1}  f. 
\]
 Therefore  the spectrum of $\widehat\A=\nu_0^{-1}\A$ is equal to the spectrum of $Q$.
    
Denote by    $B= I-Q$ where $I$ is the identity operator on $L^2(T_L,\pi)$.
We have the following obvious result.
\vskip0.5cm
\begin {lem} \label {L1} The spectrum of $B$ is equal to the spectrum of $ I - \widehat\A$, where $I$ is the identity operator 
on  $L^2(T_L, \frac 1 {p(x)} dx)$.
\end {lem}

Next we show that the spectrum of $B$ restricted to  functions     orthogonal  in $L^2(T_L,\pi )$ to the
constant functions, so that $\int f(x) \pi(x) dx =0$, is strictly positive.  The gap is bounded by    a constant independent on $L$.  
To shorten notation we denote for functions $f$ and $g$ in $L^2(T_L,\pi)$ 
\begin {equation} \label {S.3} \int_{T_L} f(x) g(x) \pi (x) {\rm d} x = ( f,g).   \end {equation}
 Denote by 
\begin {equation}  \label {OC5} \nu_1 = \inf_{\{f:(f,1)=0\}} \frac {(f,Bf)} {(f,f)}.   \end {equation}
   We will show that there exists a constant $D$ independent on $L$ so that 
$  \nu_1   \ge  D >0 $. 
We  obtain this  by  applying    \cite[Theorem 2.1] {LS} and  estimating   the Cheeger's constant. 
First notice that  the  linear bounded operator $B$ defined on  $L^2(T_L,\pi)$ can be written as
\begin {equation} \label {S1.15} 
(B g)(x) \equiv \int_{T_L}  Q (x,y) [ g(x) - g(y) ] {\rm d} y,  
\end {equation}
i.e as  the generator of a continuous time markovian jump process  with transition rate kernel $Q(\cdot, \cdot)$
and invariant probability $ \pi$.  
Define, see  \cite {LS},   the Cheeger's constant as
\begin {equation} \label {S1.13} 
\ks \equiv \inf_{A  \in  \SSS, 0 <\pi(A) <1}  \ks(A) 
\end {equation}
 where   $\SSS$ denotes the $\pi -$  measurable sets of $T_L$ and 
\begin {equation} \label {S1.14} 
\ks(A)\equiv \frac {\int \pi(x){ \rm d} x \1_A(x) \left ( \int Q (x,y)  \1_{A^c}(y)  {\rm d}  y \right ) } {
\pi(A) \pi (A^c)} =  \frac { ( \1_A,Q\1_{A^c})} {\pi(A) \pi (A^c)}.  
\end {equation}
 Taking into account that  $ ( \1_A,\1_{A^c}) =0$ we can write \eqref {S1.14} as the following:
\begin {equation} \label {S1.14a} 
\ks(A)= \frac { ( \1_A,Q\1_{A^c})  } {\pi(A) \pi (A^c)} =  -\frac { ( \1_A,B\1_{A^c}) } {\pi(A) \pi (A^c)}   = \frac { ( \1_A,B\1_{A}) } {\pi(A) \pi (A^c)}.
\end {equation}
Since 
 \begin {equation} \label {S1.18a}\int \pi (x) Q (x,y)   {\rm d}  x  = \frac 1 {\nu_0} v_0(y) (J \star_b v_0) (y) = \pi(y)  \end {equation}
   \begin {equation} \label {S1.19}\int \pi (x) Q
(x,y)    {\rm d}  y  = \frac 1 {\nu_0} v_0(x) (J \star_b v_0) (x) = \pi(x), \end {equation}
the constant $M$ appearing in \cite {LS} in our  case is simply $M=1$.   Next, we  recall   \cite [Theorem 2.1] {LS},  which    in the present context   reads:
\vskip0.5cm  
\begin {thm} (\cite {LS})  Let  $B$ be the  bounded self-adjoint operator on $L^2(T_L, \pi)$ defined in \eqref {S1.15} 
whose  marginals, in term of the invariant measure,  are  given in \eqref {S1.18a} and  \eqref {S1.19}. Then
  \begin {equation} \label {S1.20} 
  \kappa \frac {\ks^2} 8 \le \nu_1\le \ks  
  \end {equation}
where $\ks$ is defined in \eqref {S1.13}, \eqref {S1.14} and $\kappa$ is a positive constant.
\nada {\begin {equation} \label {S1.21}\kappa \equiv \inf_{\D }\sup_c \frac {\left ( \E|(X+c)^2- (Y+c)^2|\right )^2} { \E|(X+c)^2|}  \end {equation}
where
the infimum is taken over all distributions $\D$  of i.i.d. real-valued random variable $(X,Y)$ with variance
1. }
\end {thm}
  \noindent   
  The interesting and deeper  part of the previous theorem  is the lower bound of $\nu_1$.   It    states that if there does not exist  a set $A$ for which the flow from $A$ to $A^c$  is unduly  small then the Markov chain must have rapid convergence to equilibrium, or more precisely that   $B$ restricted to function orthogonal to the constant must have spectrum strictly positive. 

\begin {thm} \label {LS1} For $ \beta>1$   there exists $L_1(\b)$ so that for $L \ge L_1(\beta)$    the Cheeger's constant associated to the operator $B$ on $L^2(T_L,\pi )$, defined in \eqref {S1.13} and \eqref {S1.14}  is bounded below by a positive  constant $D$, given in \eqref {LR.1}, depending on $\beta$ and on the interaction $J$, but     
independent on $L$ so that 
 \begin {equation} \label {S1.22}  \ks\ge D.  \end {equation}
 \end {thm}
\begin {proof}  For $\beta>1$,  fix any  $ \e_0= \e_0 (\beta)$,   $ \e_0 \in    (0,  \frac {(1-\s(m_\b))} {2})$    and  take $L_1(\beta)$ so that    Lemma  \ref {S13} holds.  For any   $L\ge L_1(\beta)$ we  estimate  the Cheeger's constant,  see \eqref {S1.13}.
%Since $k(A)=k(A^c)$, is enough to consider measurable sets  $A$, such that $ \pi (A)\leq \frac 12$.  
For each measurable set $A\subset T_L$, by   definition \eqref {S1.14}  we have 
\begin {equation} \label {S1.24c} 
\begin {split} 
&  \ks(A)\equiv \frac {\int \pi(x)  {\rm d}  x \1_A(x) \left ( \int Q(x,y)\1_{A^c}(y)  {\rm d} y\right ) } {\pi(A) \pi (A^c)}    \cr 
& =    \frac 1 {\nu_0}\frac {\int \pi(x) {\rm d}  x \1_A(x) \left ( \int p(x) J (x,y) \frac {v_0(y)}{v_0(x)}   \1_{A^c}(y)  {\rm d}  y\right ) } 
{\pi(A) \pi (A^c)} \cr &   \ge    \frac 1 \gamma  \b(1-m^2_\b)   \frac {\int \pi(x) {\rm d}  x \1_A(x) \left ( \int  J (x,y)    
\1_{A^c}(y)  {\rm d}  y\right ) } {\pi(A) \pi (A^c) },  
 \end {split} 
 \end {equation}  
 where $A^c=T_L\setminus A$ and we  applied   \eqref {S.10},  \eqref {S1.8} and  Lemma \ref {S14}. 
 Define, for any $x \in \R$,\footnote{ Here we consider $\pi$ as a measure on $\R$, but supported on $T_L$.} 
 \begin {equation} \label {CE1} 
 \phi_A (x)=\begin{cases}  \frac {\pi (A \cap [x, x+1) )} {\pi ([x, x+1))}\quad&\textrm{ if }\pi ([x, x+1))\neq 0\\
 0&\textrm{ otherwise}.
 \end{cases}
 \end {equation}
The function $ \phi_A $  defined in \eqref {CE1} is the conditional measure of the set $A$ with respect to the interval $[x, x+1)$. When $x \notin \widetilde T_L=(-L-1,L)$ we have $ \phi_A(x)=0$. 
For the time being we are interested only in the case $\phi_A(0)\leq \frac 12$. If $\phi_A(0)> \frac 12$, then we will see that the same argument can be applied to the set $A^c$ (see equation \eqref{CE21a}\,) for which $\phi_{A^c}(0)\leq \frac 12$ again. Next we define two points $x^*, x^{**}$ that roughly describe the distribution of $A$ in $T_L$.

When $ \phi_A (x)< \frac 12$ for all $x\geq 0$,
 we  set  $x^*= L$ and when and $ \phi_A (x)< \frac 12$ for all $x\leq 0$, then we  set $x^{**}= -L-1$, otherwise
 \begin {equation} \label {CE3} 
x^*= \min\left\{x \ge 0: \phi_A (x)= \frac 12 \right\},  \qquad x^{**}= \max \left\{x \leq 0 : \phi_A (x)= \frac 12 \right\}.  
\end {equation}
% Notice that it cannot happen that   $ \phi_A (x)> \frac 12$ for all  $x \in \widetilde T_L$,   since  $\pi (A)\leq \frac 12$. 
 
Our goal will be to lower bound the numerator of the last term in \eqref {S1.24c} in terms of $\pi (A)$. Upper bounding  $\pi (A^c)<1$  we  immediately  obtain  a lower bound of $\ks(A)$.
%When $\phi_A(0)\geq \frac 12$, we will lower bound the numerator of the last term in \eqref {S1.24c} in term of $\pi (A^c)$.
%Upper bounding  $\pi (A)< 1$  we  immediately  have  a lower bound of $\ks(A)$. 

Let us start the lower bound. Recall that, for the moment, we are assuming $\phi_A(0)\leq \frac 12$. We use the convention that $[x^{**}+1, x^*]=\emptyset$ if $x^{**}+1> x^*$ and write\footnote{ The idea here is that the part of $A$ that we ignore is dominated by the part that we keep. This is made precise by equation \eqref{S1.32a}.}
\begin {equation} \label {CE21} 
\begin {split} 
& \int \pi(x) {\rm d}  x \1_A(x) \left ( \int  J (x,y)    \1_{A^c}(y)  {\rm d}  y\right ) \\
&  \ge     \int \pi(x) {\rm d}  x \1_{A\cap [x^{**}+1, x^*]}(x) \left ( \int  J (x,y)    \1_{A^c}(y)  {\rm d}  y\right )  \cr 
&     +  \int \pi(x) {\rm d}  x \1_{A\cap [x^*, x^*+1]}(x) \left ( \int  J (x,y)    \1_{A^c}(y)  {\rm d}  y\right )  \cr 
&        +  \int \pi(x) {\rm d}  x \1_{A\cap [x^{**}, x^{**}+1]}(x) \left ( \int  J (x,y)    \1_{A^c}(y)  {\rm d}  y\right ).  
\end {split} 
\end {equation}
We partition the first integral in the right hand side of \eqref {CE21} as follows:\footnote{ Note that the summation on the right hand side of the equation is empty if $\lfloor x^{**}\rfloor+1 > \lfloor x^*\rfloor-1$.}
\begin {equation} \label {S1.24ca} 
\begin {split} 
& \int \pi(x) {\rm d}  x \1_{A\cap [x^{**}+1, x^*]}(x) \left ( \int  J (x,y)    \1_{A^c}(y)  {\rm d}  y \right )  \cr 
&   =     \sum_{ \lfloor x^{**}\rfloor+1 \le i \le \lfloor x^*\rfloor-1}  \int \pi(x) {\rm d}  x \1_{A\cap [i,i+1)}(x) \left ( \int  J (x,y)    
\1_{A^c}(y)   {\rm d}  y\right )   \cr 
&   +     \int \pi(x) {\rm d}  x \1_{A\cap [\lfloor x^*\rfloor, x^*]}(x) \left ( \int  J (x,y)    \1_{A^c}(y)  {\rm d}  y\right )    \cr 
&   +  \int \pi(x) {\rm d}  x \1_{A\cap [x^{**}+1,\lfloor x^{**}\rfloor +1]}(x) \left ( \int  J (x,y)    \1_{A^c}(y)  {\rm d}  y\right ).
\end {split} 
\end {equation}
Each addend  of the  right hand side of the previous equation can be estimated as follows:
\begin {equation} \label {CE5} \begin {split} &
 \int \pi(x) {\rm d}  x \1_{A\cap B}(x) \left ( \int  J (x,y)    
\1_{A^c}(y)  
{\rm d}  y
\right )   \cr & \ge   \pi(A\cap B) \inf_{x\in A\cap B}  \left ( \int  J (x,y)    
\1_{A^c}(y)  
{\rm d}  y
\right ).
\end {split} \end {equation}
To bound below \eqref {CE5} note  that, for $B\in\{[x^{**}+1,\lfloor x^{**}\rfloor +1],\{[i,i+1)\}, [\lfloor x^*\rfloor, x^*]\}$, we have $ \inf_{x\in A\cap B} \int  J (x,y)    
\1_{A^c}(y)  {\rm d}  y
  \ge C>0$ independent on $L$.
This can be proved by exploiting that, for $x \in [x^{**},x^*]$, $\phi_A (x) \leq\frac 12$, which implies, by definition \eqref {CE1}, that  
\begin {equation} \label {CE5a}    
\frac {  \pi (A^c \cap [x, x+1) )} { \pi( [x, x+1))}\ge \frac 12, \qquad x \in A\cap B. 
\end {equation}
By the mean value Theorem we have that there exists $\bar x \in (x, x+1)$ so that
\begin {equation} \label {CE6}   
\int_{[x, x+1)}  \pi (z) {\rm d}  z = \pi (\bar x).  
\end {equation}
Hence, from \eqref  {CE5a}, if $x \in A\cap B$,
\begin {equation} \label {CE6a}  
\frac 12  \le  \frac {  \pi (A^c \cap [x, x+1) )} { \pi ([x, x+1))}=  \int  \frac { \pi (z)} { \pi (\bar x)} \1_{A^c \cap [x, x+1)}(z) {\rm d}  z . 
\end {equation}
By \eqref {S.10} and  \eqref {S1.8}, for $|z-\bar x| \le 1$, we obtain
\begin {equation} \label {CE7} 
\frac { \pi (z)} { \pi (\bar x)} = \frac {p(\bar x) v_0^2(z)}  {p(z) v_0^2(\bar x)}  \le  \frac   {\b}  {\b (1-m^2_\b) }    \g^2.  
\end {equation}
 Therefore 
 \begin {equation} \label {CE8}   
      |A^c \cap [x, x+1)| \ge \frac 12    (1-m^2_\b)   \frac 1 {\g^2}\equiv \omega.  
 \end {equation}
 Recalling that  $J$  is exactly supported by the unit interval (see the incipit of section \ref{sec:interaction}), we have that there exists $J_\omega>0$ so that 
\begin {equation} \label {CE12}  
 \inf_{|x-y|\leq 1- \frac \omega 2} J (x,y)   \ge J_\omega. \end {equation}   
 Moreover,  from \eqref {CE8},
  \begin{equation} \label {CE13}  
  |A^c \cap [x,x+1- \omega /2] | \ge \frac \omega 2.
  \end{equation} 
 Hence,
 \begin {equation} \label {CE14} 
  \begin {split} 
  &\inf_{x\in A\cap [i,i+1)}  \left ( \int  J (x,y)    \1_{A^c}(y)  {\rm d}  y\right )  \\
  &  \ge   \inf_{x\in A\cap [i,i+1)}  \left ( \int  J (x,y)    \1_{A^c \cap [x,x+1- \frac \omega 2)}(y)  {\rm d}  y \right )  \ge   \frac \omega 2   J_\omega =D_1.   
\end {split}  
\end {equation}    
Substituting \eqref{CE14}  in \eqref{CE5} we obtain
\begin{equation}\label{eq:carla1}
\begin{split}
\sum_{ \lfloor x^{**}\rfloor+1 \le i \le \lfloor x^*\rfloor-1}  &\int \pi(x) {\rm d}  x \1_{A\cap [i,i+1)}(x) \left ( \int  J (x,y)    
\1_{A^c}(y)   {\rm d}  y\right )\\
&\geq D_1\pi([\lfloor x^{**}\rfloor+1, \lfloor x^*\rfloor]).
\end{split}
\end{equation}
  The last  two integrals in the formula \eqref {S1.24ca} are estimated similarly.
We therefore obtain the wanted lower bound for the first term in the right hand side of \eqref{CE21}:
\begin {equation} \label {CE10} 
 \int \pi(x) {\rm d}  x \1_{A\cap [x^{**}+1, x^*]}(x) \left ( \int  J (x,y)    \1_{A^c}(y)  {\rm d}  y\right )    
 \ge   D_1     \pi (A\cap [x^{**}+1, x^*]).    
\end {equation}
The  last two terms in \eqref {CE21} are non zero only if $x^*<L$ and $x^{**}>-L-1$, respectively.\footnote{ Note that if $x^*=L$ and $x^{**}=-L-1$, then $\pi (A\cap [x^{**}+1, x^*])=\pi(A)$ whereby yielding \eqref{S1.32} without further ado.}  In such a case they can be estimated taking into account that, by the definition of $x^*$  and $x^{**}$, $\phi_A (x^*)=\phi_A (x^{**}) = \frac 12$.  This  implies that   $  \pi (A \cap [x^*, x^*+1) ) = \frac 12 \pi ([x^*, x^*+1))$, hence
$  \pi (A^c \cap [x^*, x^*+1) ) = \frac 12 \pi ([x^*, x^*+1))$ and  as we did before, we obtain
\begin {equation} \label {CE11} \begin {split} & \int \pi(x) {\rm d}  x \1_{A\cap [x^*, x^*+1]}(x) \left ( \int  J (x,y)    
\1_{A^c}(y)  
{\rm d}  y
\right )    
\ \cr &   \ge  \int \pi(x) {\rm d}  x \1_{A\cap [x^*, x^*+1]}(x) \left ( \int  J (x,y)    
\1_{A^c\cap [x^*, x^*+1]}(y)  
{\rm d}  y
\right )  
  \cr &   \ge  D_1
 \pi (A\cap [x^*, x^*+1]) =    \frac {D_1}2
 \pi ( [x^*, x^*+1]).     \end {split} \end {equation}
 The same estimate holds for the last term of  \eqref  {CE21}.
 \begin {equation} \label {CE12a}  
  \int \pi(x) {\rm d}  x \1_{A\cap  [ x^{**},x^{**}+1]}(x) \left ( \int  J (x,y)    \1_{A^c}(y)  {\rm d}  y\right )  
    \ge    \frac {D_1}2  \pi ( [ x^{**},x^{**}+1]).     
    \end {equation}
  Taking into account  \eqref{CE21}, \eqref{CE10},  \eqref{CE11},  \eqref{CE12a} we obtain 
  \begin {equation} \label {S1.24ba} 
  \begin {split} 
  & \int \pi(x) {\rm d}  x \1_{A}(x) \left ( \int  J (x,y)  \1_{A^c}(y)  {\rm d}  y\right )  \cr 
  &    \ge     D_1     \pi (A\cap [x^{**}+1, x^*])+\frac {D_1}2 \pi ( [x^*, x^*+1])+ \frac {D_1}2 \pi ( [ x^{**},x^{**}+1]) \cr 
  & \ge  \frac {D_1}2  \bigg \{  \pi (A\cap [x^{**}+1, x^*])+        \pi ( [x^*, x^*+1])+   \pi ( [ x^{**},x^{**}+1]) \bigg \}.
\end {split} \end {equation}
To conclude the analysis of this case we must upper bound   $\pi(A \cap  [x^*,L])$ in terms of $ \pi([x^*,x^*+1])$ and   $\pi(A \cap  [-L,x^{**}+1])$ in terms of $\pi([x^{**},x^{**}+1] )$.   We have 
\begin {equation} \label {S1.27a}   
\begin {split} \pi(A \cap [x^*, L]) &= \int_{A \cap [x^*, L]} \frac {v^2_0(x)} {p(x)} {\rm d}  x  \le \frac 1 {\b (1-m^2_\b)} \int_{x^*}^L  
v^2_0(x){\rm d}  x  \\
&\le  \frac 1 {\b (1-m^2_\b)} \sum_{k\ge 0} \int_{x^*+k}^{x^*+k+1}   v^2_0(x){\rm d}  x   \\
 &=  \frac 1 {\b (1-m^2_\b)} \sum_{\{k\geq 0\;:\;   x^*+k\le r_0 \} }
\int_{x^*+k}^{x^*+k+1}   v^2_0(x){\rm d}  x \\
&\phantom{=}+ \frac 1 {\b (1-m^2_\b)} \sum_{\{k\ge 0\;:\; x^*+k > r_0\} } 
\int_{x^*+k}^{x^*+k+1}   v^2_0(x){\rm d}  x.  
 \end {split} 
 \end {equation}
 When $x^*+k\leq r_0$ we can write
\begin {equation} \label {S1.28b}   
\begin {split}   \int_{x^*+k}^{x^*+k+1}     v^2_0(x) {\rm d} x  &=\int_{x^*+k-1}^{  x^*+k}   v^2_0(x+
1){\rm d}  x   =\int_{x^*+k-1}^{x^*+k}    v^2_0(x)  \frac { v^2_0(x+1)}{ v^2_0(x)} {\rm d} x\\
& \le  \g^2 \int_{x^*+k-1}^{x^*+k}     v^2_0(x)   \le    (\g^2)^k \int_{x^*}^{x^*+1}     v^2_0(x) {\rm d} x, 
\end{split}
 \end {equation}
where $\g>1$, see   \eqref {S1.8}.
 When  $x^*+k> r_0$, by \eqref{S1.6} and \eqref{S1.8},   we have
\begin {equation} \label {S1.28c} 
\begin{split}
\int_{x^*+k}^{x^*+k+1}   v^2_0(x){\rm d}  x&\leq \gamma^2 e^{-2(k-k(r_0))\a (\e_0)}\int_{r_0}^{r_0+1}   v^2_0(x){\rm d}  x\\
&=:\gamma^2 d_1^{2(k-k(r_0))}\int_{r_0}^{r_0+1}   v^2_0(x){\rm d}  x
\end{split}
\end {equation}
where $k(r_0)$ is the first integer $k$  so that  $x^*+k(r_0)> r_0$, note that  $k_0\leq \lfloor r_0\rfloor+1$.
Remembering  \eqref  {S1.28b},  we obtain
\begin {equation} \label {S1.31} 
\begin {split}  
& \int_{x^*+k}^{x^*+k+1}   v^2_0(x) {\rm d} x \le\gamma^2 (d_1^2)^{k-k(r_0)}   \int_{x^*+k(r_0)}^{x^*+k(r_0)+1}  v^2_0(x) {\rm d} x  \cr & \le    (\g^2)^{k(r_0)+2}  (d_1^2)^{k-k(r_0)}   \int_{x^*}^{x^*+1}  v^2_0(x) {\rm d} x.
\end {split} 
 \end {equation}
Therefore, from \eqref {S1.27a},  \eqref {S1.28b}, \eqref  {S1.31} and  denoting
\[ 
D_2=   \frac 1 {(1-m^2_\b)}  \left \{ \sum_{k\leq \lfloor r_0\rfloor+1 } (\g^2)^k +  (\g^2)^{ \lfloor r_0\rfloor+3} \sum_{k\geq 0}  (d_1^2)^{k}  \right \},
\]
we obtain
\begin {equation} \label {S1.32a} \begin {split}  &   \pi (A\cap [ x^*,L]) \le      
D_2   \pi ([x^*,x^*+1]).
\end {split} \end {equation}
Similar computations hold for the interval $[x^{**},x^{**}+1]$ if $x^{**}>-L-1$.
We can then  continue the estimate left at \eqref {S1.24ba} and write
\begin {equation} \label {S1.32} 
\begin {split}  
& \frac {D_1}2  \left \{  \pi (A\cap [x^{**}+1, x^*])+   \pi ( [x^*, x^*+1])+    \pi ( [ x^{**},x^{**}+1]) \right \}  \cr 
&\ge \frac {D_1}2  \pi (A\cap [x^{**}+1, x^*]) + \frac {D_1}{2 D_2} \left \{ \pi (A\cap [ x^*,L]) + \pi (A\cap [-L, x^{**}+1]) \right \} \\
&\ge \min \left  \{ \frac {D_1}2, \frac {D_1} {2 D_2}\right \}  \pi (A). 
    \end {split} 
 \end {equation}
Collecting the above facts and inserting them in \eqref {S1.24c}  we obtain that, in the case $\phi_A (0) \leq\frac 12$, 
\begin {equation} \label {S1.34}   
\ks(A) \ge    \frac   {\b (1-m^2_\b)} \gamma   \min \left  \{ \frac {D_1}2, \frac {D_1}{2 D_2}\right \}. 
 \end {equation}

 To conclude we must lower bound  the numerator of the last term in \eqref {S1.24c} when $\phi_A(0)\geq\frac 12$.  Exploiting that  the interaction $J$ is symmetric  and that, for $|x-y|\leq 1$, $\frac {\pi (x)} {\pi (y)} \ge   \frac 1 {\g^2}  (1-m^2_\b)$ we have 
\begin {equation} \label {CE21a} 
\begin {split} & \int \pi(x) {\rm d}  x \1_A(x) \left ( \int  J (x,y)    
\1_{A^c}(y)  {\rm d}  y \right ) \\
& = \int {\rm d}  y \1_{A^c}(y) \left ( \int  J (y,x)    \1_{A}(x)  \pi(x) {\rm d}  x
\right )  \\
 & \ge  \frac 1 {\g^2}  (1-m^2_\b) \int {\rm d}  y \pi(y)  \1_{A^c}(y) \left ( \int  J (y,x)    
\1_{A}(x)   {\rm d}  x
\right ). 
\end {split} \end {equation}
We have thus obtained the same expression present in \eqref{S1.24c} only with the roles of the sets $A^c$ and $A$ exchanged. Accordingly, we can proceed exactly as before, since $\phi_{A^c}(0)=1-\phi_A(0)\leq\frac 12$, obtaining
\begin{equation}\label{eq:carla2}
 \int \pi(x) {\rm d}  x \1_{A^c}(x) \left ( \int  J (x,y)  \1_{A}(y)  {\rm d}  y\right ) \geq \min \left  \{ \frac {D_1}2, \frac {D_1} {2 D_2}\right \}  \pi (A^c).
 \end{equation}
Thus, by \eqref {S1.24c}, \eqref {CE21a} and \eqref{eq:carla2} we have, when $\phi_A (0)\geq \frac 12$,
\begin {equation} \label {S1.34-z}   \ks(A) \ge     \frac {(1-m^2_\b) } {\g^2}    \frac   {\b (1-m^2_\b)} \gamma   \min \left  \{ \frac {D_1}2, \frac {D_1}{2 D_2}\right \}. 
 \end {equation}
Accordingly, denoting  
 \begin {equation} \label {LR.1} 
 D =   \min \left \{ \frac   {\b (1-m^2_\b)} \gamma  ,    \frac   {\b (1-m^2_\b)^2}{\g^3 }  \right \} \min\left \{ \frac {D_1}2, \frac {D_1}{2 D_2}\right \}
  \end {equation}
and remembering \eqref{S1.34} and \eqref{S1.34-z}, we obtain $\ks(A) \ge D$.
 The thesis follows. 
 \end {proof}

  { \bf  Proof of Theorem \ref {81}}
 For $\beta>1$,  fix any $ \e_0 = \e_0(\beta)$, $ \e_0  \in    (0,  \frac {(1-\s(m_\b))} {2})$  and take 
  $L_1(\beta)$ so that     Lemma  \ref {S13}  and  Theorem  \ref {LS1} hold.  
Recall that   $ \LL^0 = I - \A$, where $I$ is the identity operator and   $\A$ is the  operator 
defined in \eqref {S1.2}.  By Theorem \ref {P-F}  we  have  immediately  that  $ \LL^0 $ is   a bounded, selfadjoint,    quasi compact operator.   
The  smallest  eigenvalue of  $ \LL^0$ is $\mu_1^0= 1- \nu_0$      where $\nu_0$    is   the maximum eigenvalue of $\A$ and    $\psi_1^0= v_0$  is the corresponding eigenfunction.  Equation \eqref {8.2}  follows since  Lemma \ref {S11} and Lemma \ref {S14} state $1>\nu_0\geq 1-ce^{-2\alpha L}$.
Point (2) is a direct consequence of Lemma \ref {L1}, \eqref {S1.20} and Theorem \ref {LS1}.
Next we show point (3). 
  Split 
\begin{equation} \label {anc1}     \frac {\bar m'}{\|\bar m'\|} =   a \psi^0_1  + (\psi^0_1  )^{ort}. \end {equation} 
  Then
\begin{equation} \label {8.31}  a^2  +  \|(\psi^0_1 
)^{ort}  \|^2   =1  \end {equation} 
\begin{equation} \label {8.30}   \frac 1   {\|\bar m'\|^2 } \langle \LL^0  \bar m', \bar m' \rangle = a^2 \mu_1^0 +  \langle \LL^0   (\psi^0_1  )^{ort},
 (\psi^0_1  )^{ort} \rangle \ge a^2 \mu_1^0 + \mu_2^0  \|(\psi^0_1  )^{ort} \|^2. \end {equation} 
By  Lemma \ref  {S11}   
$$  \frac 1   {\|\bar m'\|^2 } \langle \LL^0  \bar m', \bar m' \rangle  \le    c  e^{-  2\a L}        $$ 
 hence  from \eqref {8.31} and \eqref   {8.30}
$$c  e^{-  2\a L}  \ge  (1-  \|(\psi^0_1 
)^{ort}  \|^2)  \mu_1^0 + \mu_2^0  \|(\psi^0_1  )^{ort} \|^2. $$
By \eqref {8.2} and \eqref {8.3},    that  there exists a $C>0$ independent on $L$ so that 
\begin{equation} \label {8.31a}  \|(\psi^0_1  )^{ort} \|^2  \le  C e^{- 2\a L}.  \end {equation}
Then  \eqref {8.4} follows by \eqref {anc1}, \eqref {8.31} and \eqref {8.31a}.
\qed

 \end {document}